\documentclass[pra,superscriptaddress,showpacs,twocolumn,10pt]
{revtex4}
\pdfoutput=1
\usepackage{amsmath,amsthm}
\usepackage{amssymb,latexsym}
\usepackage{float}
\usepackage{graphicx}

\newtheorem{theorem}{Theorem}
\newtheorem{protocol}[theorem]{Protocol}
\newtheorem{lemma}[theorem]{Lemma}

\theoremstyle{definition}
\newtheorem{definition}[theorem]{Definition}
\newtheorem{remark}[theorem]{Remark}

\floatstyle{ruled}
\newfloat{protocol}{htb}{loa}
\floatname{protocol}{Protocol}

\newcommand{\qedsymb}{\hfill{\rule{2mm}{2mm}}}
\renewenvironment{proof}[1][]{\begin{trivlist} % Note changed to renewenvironment by Ben because I loaded amsthm package.
\item[\hspace{\labelsep}{\bf\noindent Proof#1:\/}] }{\qedsymb\end{trivlist}}

\newcommand{\be}{\begin{equation}}
\newcommand{\ee}{\end{equation}}

\newcommand{\ba}{\begin{array}}
\newcommand{\ea}{\end{array}}

\newcommand{\bc}{\begin{cases}}
\newcommand{\ec}{\end{cases}}

\newcommand{\beba}{\begin{equation}\begin{array}{lll}}
\newcommand{\eeea}{\end{array}\end{equation}}

\newcommand{\ssll}[2]{\sum\limits_{#1}^{#2}}
\newcommand{\ppll}[2]{\prod\limits_{#1}^{#2}}

\newcommand{\zp}{\mathbb{Z}_p}

\newcommand{\zpn}{\mathbb{Z}^n_p}
\newcommand{\zpm}{\mathbb{Z}^m_p}
\newcommand{\zpN}{\mathbb{Z}^N_p}

\renewcommand{\log}{\mathrm{log}_2}

\newcommand{\tpi}[1]{\Theta^{-1}_{p}(#1)}
\newcommand{\tp}[1]{\Theta_{p}(#1)}

\newcommand{\vx}{\vec{x}}
\newcommand{\vy}{\vec{y}}
\newcommand{\vz}{\vec{z}}

\newcommand{\too}{\leftrightarrow}

\begin{document}

\title{Functional boxes, communication complexity and information causality}

\author{Guoming Wang}
\email{gmwang@eecs.berkeley.edu}
\affiliation{Computer Science Division, University of California Berkeley, Berkeley, California 94720, USA}

\date{\today}

\begin{abstract}
We propose a class of nonlocal boxes named functional boxes which include a generalization of the Popescu-Rohrlich(PR) box as a special case. We show that every functional box corresponding to an additively inseparable function can make communication complexity trivial and thus seems unlikely to exist in nature. Then we relate general nonlocal boxes to these functional boxes, and derive several limits on them from the principle of information causality.
\end{abstract}


\maketitle

\section{Introduction}

%why not maximally nonlocal
Quantum mechanics allows distant parties to establish certain correlations that are impossible in the classical world. These correlations are nonlocal, in the sense that they violate Bell inequality \cite{Bel64,CHSH69}. However, as first shown by Tsirelson\cite{Tsi80}, their violation is still bounded. In a seminal paper, Popescu and Rohrlich \cite{PR94} proposed a hypothetical nonlocal correlation that attains the maximal value for the Clauser-Horne-Shirmony-Holt(CHSH) inequality\cite{CHSH69} but still cannot be used to signal from one party to another. Understanding why such stronger-than-quantum correlations have not been observed in nature has since then become an active topic of research \cite{Dam05,WW05,BP05,BLM+05,BCU+06,BBL+06,MAG06,SGB+06,BB+07,LPS+07,Bar07,MR08,ABL+09,PPK+09,BS09,SBP09,ABP+09,NW09,Hsu09,GMC+10,
OW10,ABB+10,CSS10,XR11,Hsu11,GWA+11,SS11}. Recently, several information-theoretical principles were proposed as candidates that separate physically realizable correlations from nonphysical ones. In this paper, we will be concerned with two of them --- nontrivial communication complexity \cite{Dam05} and information causality \cite{PPK+09}. Both of them show that allowing certain postquantum correlations would lead to implausible simplification of communication tasks.

%limits from communication complexity
In one line, van Dam\cite{Dam05} showed that, equipped with many copies of Popescu-Rohrlich(PR) boxes\cite{PR94}, every communication complexity problem can be solved deterministically by a single bit of communication. Then, Brassard \emph{et al.}\cite{BBL+06} extended this result by proving that if PR boxes can be implemented with probability greater than $(3+\sqrt{6})/6 \approx 90.8\%$, then the probabilistic communication complexity also collapses. These results were further generalized by Refs.\cite{BP05,MR08,BS09}.

%limits from information causality
In another line, Paw{\l}owski \emph{et al.} \cite{PPK+09} considered a scenario in which Alice had a database of $N$ independent and random bits, and a distant party, Bob, was asked to guess a the $k$-th bit in Alice's database for random $k \in \{0,1,\dots,N-1\}$. They suggested a principle stating that Bob can gain at most $m$ bits of information about Alice's database by using his local resources and receiving $m$ bits from Alice. This principle was named information causality. When $m=0$, it reduces to no-signalling. It was demonstrated that using many copies of PR boxes Bob can correctly guess any bit of Alice with certainty by receiving only $1$ bit from her. Moreover, any correlation exceeding Tsirelson's bound for the CHSH inequality violates information causality. The implications of information causality on nonlocality were further explored in Refs.\cite{ABP+09,Hsu09,CSS10,XR11,Hsu11,GWA+11,SS11}.

%our work
A common feature of these two proposals is that the PR box perfectly cracks the information processing task under consideration. It also exhibits an extremely strong power for other tasks \cite{WW05,LPS+07}. The reason for the success of this box can be summarized as follows. Essentially, the aforementioned tasks can be viewed as that Alice and Bob want to compute some function $f(x,y)$ in some distributed way, where $x$ and $y$ are initially held by Alice and Bob respectively. Note that any boolean (or arithmetic) function can be computed by a circuit consisting of XOR (or addition) and AND (or multiplication) gates. These two types of gates generally do not commute. However, using a PR box, AND operations are ``transformed" into XOR operations,
in the way that the AND of two inputs is encoded as the XOR of two outputs. So Alice and Bob effectively convert the original circuit into a ``circuit" consisting of only XOR operations, which are commutative. Then Alice and Bob can reorder their operations, compact them, and minimize the interaction between them. This argument was made explicit in the discussion of communication complexity\cite{Dam05}. But it can also explain the success of PR box for the information causality task, since that task can be viewed as the communication complexity problem for a special function -- the index function $f(\vx,y)\equiv x_y$, with the extra condition that the communication is only from Alice to Bob.

In light of the above observation, we introduce a class of nonlocal boxes called \textit{functional boxes}. This class of boxes include a generalization of PR box as a special case. A functional box's inputs and outputs
have $p$ possible values which are identified with the elements of $\zp$, where $p$ can be any prime number. This box encodes the value of a function $f:\zp \times \zp \to \zp$ over its inputs as the difference of its outputs. In particular, the PR box is just the functional box corresponding to $f(x,y)=xy$ and $p=2$. We find that, as long as $f$ is not \textit{additively separable} (see section 2 for definition), the corresponding functional box is asymptotically equivalent to a generalized PR box. (More specifically, an additively inseparable function contains a component that is about the multiplication of its two variables, and we show that this component can be isolated out by using a difference method, and hence the corresponding functional box can be transformed into a generalized PR box.) We also prove that a generalized PR box can enable perfect distributed computation and make communication complexity trivial. As a result, any functional box equivalent to it has the same power and is unlikely to exist in nature. Then the next question is how well these functional boxes might be approximated. We derive several bounds on the proximity between a plausible $p$-nary-input, $p$-nary-output box and these functional boxes, from the principle of information causality. In order to do this, we extend the basic and nested protocols in Ref.\cite{PPK+09} to the $p$-nary digit case. Besides, we also present a generalization of the \textit{depolarization} process that transforms any binary-input binary-output box into an \textit{isotropic} one\cite{MAG06,Sho08}, which might be of independent interest.

%The remainder of this paper is organized as follows. In section 2, we introduce functional boxes and prove that every functional box corresponding to an additively inseparable function can make communication complexity trivial. In section 3, we relate general boxes to functional boxes and derive limits on them from the principle of information causality. Section 4 concludes this paper.

\section{Functional Boxes and Communication Complexity}

Let us first review several basic definitions about nonlocal boxes.

\begin{definition}
A bipartite correlation box (or simply a box) is a hypothetical device shared by two spatial separated parties Alice and Bob that receives an input $x \in \mathcal{X}$ from Alice and an input $y \in \mathcal{Y}$ from Bob, and outputs $a \in \mathcal{A}$ to Alice and $b \in \mathcal{B}$ to Bob, according to a joint probability distribution $P(a,b|x,y)$. Without causing ambiguity, we also call this box $P$.

If $P$ satisfies
\beba
\ssll{b}{}P(a,b|x,y)&=\ssll{b}{}P(a,b|x,y')&\equiv P(a|x),~~\forall a,x,y,y';\\
\ssll{a}{}P(a,b|x,y)&=\ssll{a}{}P(a,b|x',y)&\equiv P(b|y),~~\forall b,x,x',y,
\eeea
then it is no-signalling. Namely, Alice cannot signal to Bob via her choice of input to this box and vice versa.

If Alice and Bob can simulate $P$ using shared randomness (without communication between them), then $P$ is local. Otherwise, it is nonlocal.
\end{definition}

For example, the standard PR box is given by
\beba
PR(a,b|x,y)=
\bc
\dfrac{1}{2},~~~~\textrm{if}~~a \oplus b=x \wedge y\\
0,~~~~\textrm{otherwise}
\ec,
\eeea
where $a,b,x,y\in \{0,1\}$. Note that $PR(a|x,y)=PR(b|x,y)=\dfrac{1}{2}$,
$\forall a,b,x,y$. So this box is no-signalling.

In this paper, we focus on no-signalling boxes for which $|\mathcal{X}|=|\mathcal{Y}|=|\mathcal{A}|=|\mathcal{B}|=p$,
where $p$ can be any prime number. Without loss of generality, we assume $\mathcal{X}=\mathcal{Y}=\mathcal{A}=\mathcal{B}
=\{0,1,\dots,p-1\} \equiv \zp$, and \textit{all the following computation related to} $x,y,a,b$ \textit{is modulo} $p$.

\begin{definition}
A function $F:\zpn \times \zpm \to \zp$ is distributedly computed by Alice and Bob if, when Alice is given any $\vx=(x_0,x_1,\dots,x_{n-1}) \in \zpn$ and Bob is given any $\vy=(y_0,y_1,\dots,y_{m-1}) \in \zpm$, Alice can produce $a \in \zp$ and Bob can produce $b \in \zp$ such that $a-b=F(\vx,\vy)$.
\end{definition}

If $F$ can be distributedly computed with certainty (or with constant probability), then the deterministic (or probabilistic) communication complexity of $F$ becomes trivial, since Alice can simply send $a$ to Bob and then Bob can calculate $a-b=F(\vx,\vy)$.

Let us first consider the following box, which is a straightforward generalization of standard PR box to the $p$-nary input/output case:
\begin{definition}
\beba
PR_p(a,b|x,y)=
\bc
\dfrac{1}{p},~~~~\textrm{if}~~a-b=xy\\
0,~~~~\textrm{otherwise}
\ec,
\eeea
where $a,b,x,y \in \zp$.
\end{definition}
In particular, the standard PR box is just $PR_2$. Note that $PR_p(a|x,y)=PR_p(b|x,y)=\dfrac{1}{p}$, $\forall a,b,x,y$. So $PR_p$ is no-signalling.

Given arbitrarily many copies of $PR_p$, Alice and Bob will be able to distributedly compute any function $F:\zpn \times \zpm \to \zp$ perfectly. To prove this, fist note that $F(\vx,\vy)$ can always be written as a multivariate polynomial whose degree in each $x_i$ or $y_j$ is no larger than $p-1$:
\beba
F(\vx,\vy)
&=& \ssll{\alpha_0,\dots,\alpha_{n-1}=0}{p-1}
\ssll{\beta_0,\dots,\beta_{m-1}=0}{p-1}
\mu_{\vec{\alpha},\vec{\beta}} \prod\limits_{i=0}^{n-1}x_i^{\alpha_i}
\prod\limits_{j=0}^{m-1} y_j^{\beta_j}
\eeea
for some $\mu_{\vec{\alpha},\vec{\beta}} \in \zp$,
where $\vec\alpha=(\alpha_0,\alpha_1,\dots,\alpha_{n-1})$
and $\vec\beta=(\beta_0,\beta_1,\dots,\beta_{m-1})$. So Alice and Bob can execute the following protocol: for each $(\vec\alpha, \vec\beta)$, they use a $PR_{p}$ as follows: Alice inputs $\prod\limits_{i=0}^{n-1}x_i^{\alpha_i}$ and Bob inputs $\prod\limits_{j=0}^{m-1} y_j^{\beta_j}$, and suppose they
get outputs $a_{\vec{\alpha},\vec{\beta}}$ and
$b_{\vec{\alpha},\vec{\beta}}$ respectively. In the end, Alice sets
$a=\ssll{\vec{\alpha}}{}\ssll{\vec{\beta}}{}\mu_{\vec{\alpha},\vec{\beta}}a_{\vec{\alpha},\vec{\beta}} $ as her final output, and Bob sets
$b=\ssll{\vec{\alpha}}{}\ssll{\vec{\beta}}{}\mu_{\vec{\alpha},\vec{\beta}}b_{\vec{\alpha},\vec{\beta}} $
as his final output. It is easy to verify $a-b=F(x,y)$.

Now let us consider a wider class of boxes:
\begin{definition}
For any function $f:\zp \times \zp \to \zp$, the functional box corresponding to $f$ is defined as
\be
P^f(a,b|x,y)=
\bc
\dfrac{1}{p},~~~~\textrm{if}~~a-b=f(x,y)\\
0,~~~~\textrm{otherwise}
\ec,
\ee
where $a,b,x,y \in \zp$.
\end{definition}
Namely, $P^f$ can be directly used to distributedly compute $f$. In particular, $PR_p$ can be viewed as the functional box corresponding to $f(x,y)=xy$. Note that $P^f(a|x,y)=P^f(b|x,y)=\dfrac{1}{p}$, $\forall a,b,x,y$. So $P^f$ is also no-signalling.

\begin{definition}
A bivariate function $f:\zp \times \zp \to \zp$ is additively separable if there exist univariate functions $g,h: \zp \to \zp$ such that $f(x,y)=g(x)+h(y)$, $\forall x,y \in \zp$. Otherwise, $f$ is additively inseparable.
\end{definition}

\begin{definition}
Suppose $f:\zp \times \zp \to \zp$ can be written as
$f(x,y)=\ssll{i,j=0}{p-1} \lambda_{i,j}x^iy^j$ for some $\lambda_{i,j} \in \zp$. Define
\beba
\Delta(f)=\max\limits_{1 \le i,j \le p-1}\{\mathbb{I}_{\lambda_{i,j} \neq 0}(i+j)\},
\label{eq:deltaf}
\eeea
where $\mathbb{I}$ is the indicator function (i.e. $\mathbb{I}_E=1$ if $E$ is true, and $0$ otherwise). Namely, $\Delta(f)$ is the maximum of the degrees of the terms in $f$ that can be divided by $xy$.
\end{definition}

Obviously, a function $f$ is additively inseparable if and only if $\Delta(f) \ge 2$. Only an additively inseparable $f$ contains a term that is related to the product of $x$ and $y$. So one may naturally wonder if $P^f$ can be used to simulate $PR_p$. If so, then $P^f$ can also benefit distributed computation. We find that it is indeed the case, provided we are given sufficiently many copies of $P^f$.

\begin{definition}
We use the notation $P_1 \to P_2$ to denote the fact that we can use $N$ copies of $P_1$ to simulate a $P_2$ exactly, for some $N \ge 1$. We also use $P_1 \too P_2$ to denote that $P_1 \to P_2$ and $P_2 \to P_1$, i.e. $P_1$ and $P_2$ are asymptotically interconvertible.
\end{definition}

\begin{lemma}
If $f:\zp \times \zp \to \zp$ is additively separable, then
$P^f$ is local. Otherwise, $P^f \too PR_p$.
\label{lem:pf}
\end{lemma}

\begin{proof}
(1)Suppose $f(x,y)=g(x)+h(y)$ for some $g,h:\zp \to \zp$,
then we can build a local model for $P^f$ as follows:
Alice and Bob first generate a uniformly random variable $z \in \zp$, then Alice outputs $a=g(x)+z$ and Bob outputs $b=-h(y)+z$.

(2)Suppose $f$ is additively inseparable. We first show $PR_p \to P^f$. Recall that we have already shown how to use many copies of $PR_p$ to perform distributed computation of any functions, including $f$. So Alice and Bob can first use that protocol to obtain $a$ and $b$ such that $a-b=f(x,y)$. Then they generate a uniformly random $z \in \zp$ and modify their outputs by $a \to a+z$ and $b \to b+z$.

It remains to show $P^f \to PR_p$. Proof by induction on $\Delta(f)$:

\begin{itemize}
\item Base case: $\Delta(f)=2$. In this case, we have
\beba
f(x,y)=\lambda xy+g(x)+h(y)
\eeea
for some $\lambda \neq 0$ and $g,h:\zp \to \zp$.
We can use a $P^f$ to simulate a $PR_p$ as follows:
Alice inputs $x$ and Bob inputs $y$ to $P^f$, and suppose they obtain outputs $a$ and $b$ respectively. Then Alice sets 
\beba
a'=\lambda^{-1}(a-g(x))
\eeea
as her final output, while Bob sets 
\beba
b'=\lambda^{-1}(b+h(y))
\eeea
as his final output. Then we have
\beba
a'-b'&=&\lambda^{-1}(a-b-g(x)-h(y))\\
&=&\lambda^{-1}(f(x,y)-g(x)-h(y))\\
&=&xy.
\eeea
Furthermore, since $a$ (or $b$) is uniformly random, $a'$ (or $b'$) is also uniformly random conditioned on any $(x,y)$.

\item Inductive step: Suppose $P^f \to PR_p$ for any $f$ with $\Delta(f)=k$ for some $k \ge 2$.
Consider any $f$ with $\Delta(f)=k+1 \ge 3$. Such $f(x,y)$ (viewed as a bivariate polynomial) contains a term that is a multiple of $x^2y$ or $xy^2$. We deal with the two cases separately.

In the first case, consider
\beba
f^{(1)}(x,y) \equiv f(x+1,y)-f(x,y).
\eeea
Compared to $f$, the degree of $f^{(1)}$ in $x$ is decreased by $1$, while its degree in $y$ is the same or smaller.
Moreover, it is easy to see
\beba
\Delta(f^{(1)})=\Delta(f)-1=k \ge 2.
\eeea
Thus by induction $P^{f^{(1)}} \to PR_p$.
Furthermore, $P^{f^{(1)}}$ can be simulated with two copies of $P^f$ as follows: for the first $P^f$, Alice inputs $x+1$ and Bob inputs $y$, and assume they receive outputs $a_1$ and $b_1$ ; for the second $P^f$, Alice inputs $x$ and Bob inputs $y$, and assume
they receive outputs $a_2$ and $b_2$. Then Alice sets $a=a_1-a_2$
as her final output, and Bob sets $b=b_1-b_2$
as his final output. Then we have
\beba
a-b&=&(a_1-a_2)-(b_1-b_2)\\
&=&(a_1-b_1)-(a_2-b_2)\\
&=&f(x+1,y)-f(x,y)\\
&=&f^{(1)}(x,y).
\eeea
So $P^f \to P^{f^{(1)}} \to PR_p$.

A similar argument holds for the second case. But in this case, we consider
\beba
f^{(2)}(x,y) \equiv f(x,y+1)-f(x,y),
\eeea
which satisfies
\beba
\Delta(f^{(2)})=\Delta(f)-1=k \ge 2.
\eeea
Then by induction $P^{f^{(2)}} \to PR_p$.
Furthermore, $P^{f^{(2)}}$ can also be simulated with two copies of $P^f$. Therefore we have $P^f \to P^{f^{(2)}} \to PR_p$.

\begin{remark}
This proof actually yields a recursive strategy to simulate $PR_p$ with $2^{\Delta(f)-2}$ copies of $P^f$. To be specific, we have shown
\beba
P^f \to P^{f_1} \to P^{f_2} \to \dots \to P^{f_L} \to PR_p,
\eeea
where $f_{i}(x,y)=f_{i-1}(x+1,y)-f_{i-1}(x,y)$ or $f_{i}(x,y)=f_{i-1}(x,y+1)-f_{i-1}(x,y)$. The strategy is to use a $P^{f_L}$ to simulate a $PR_p$, where this $P^{f_{L}}$ is simulated with two copies of $P^{f_{L-1}}$, where each copy of $P^{f_{L-1}}$
is simulated with two copies of $P^{f_{L-2}}$,
and so on. Overall, $2^{\Delta(f)-2}$ copies of $P^f$ is used, since $\Delta(f_{i-1})-\Delta(f_{i})=1$ and $\Delta(f_L)=2$.

This strategy can be demonstrated by the following example. Suppose $p=3$ and $f:\mathbb{Z}_3 \times \mathbb{Z}_3 \to \mathbb{Z}_3$
is given by
\beba
f(x,y)=x^2y^2+2xy^2+xy+2x.
\eeea
Using two copies of $P^f$, we can simulate
a $P^{f_1}$ where
\beba
f_1(x,y) & = & f(x+1,y)-f(x,y)\\
&=&2xy^2+y+2.
\eeea
Then using two copies of $P^{f_1}$ (each of
which is simulated with two copies of $P^f$), we can simulate a $P^{f_2}$ where
\beba
f_2(x,y) & = & f_1(x,y+1)-f_1(x,y)\\
&=&xy+2x+1
\eeea
Finally, we can simulate a $PR_p$ with a $P^{f_2}$ by Alice subtracting her output by $2x+1$. Overall, four copies of $P^f$ is used to to simulate a $PR_p$.
\end{remark}
\end{itemize}
\end{proof}

Hence there are only two inequivalent classes of functional boxes with respect to asymptotic transformation: those corresponding to additively separable functions are local and cannot benefit distributed computation, while the others are all equivalent to $PR_p$ and can enable perfect distributed computation and make communication complexity trivial. Therefore, we have

\begin{theorem}
In any world where communication complexity is not trivial,
the functional box $P^f$ corresponding to any additively inseparable function $f:\zp \times \zp \to \zp$ cannot be implemented perfectly.
\end{theorem}

\section{Limits on Nonlocality from Information Causality}

In the previous section, we have shown that an exact implementation of $P^f$ for any additively inseparable function $f$ is impossible, unless communication complexity collapses. So now the question is how well these boxes might be implemented in nature. In this section, we partially answer this question by deriving several bounds on the proximity between any plausible $p$-nary-input, $p$-nary-output box and these functional boxes, from the principle of information causality. But before doing that, we need to first present the following result which is a key ingredient to our analysis.

\subsection{Generalized Depolarization Process}

It is well known that any binary-input, binary-output box can be transformed into an \textit{isotropic} one:
\beba
P_{iso}(\lambda) & \equiv & \lambda PR_2 + (1-\lambda) P_{N}
\\&=& \dfrac{1+\lambda}{2}PR_2 + \dfrac{1-\lambda}{2}\overline{PR_2},
\eeea
where $P_N$ is the completely random noise
(i.e. $P_N(a,b|x,y)=\dfrac{1}{4}$, $\forall a,b,x,y \in \mathbb{Z}_2$), and $\overline{PR_2}$ is the anti-PR box
(i.e. $\overline{PR_2}(a,b|x,y)=\dfrac{1}{2}$ if
$a+b=xy+1$, and $0$ otherwise,
$\forall a,b,x,y \in \mathbb{Z}_2$), via the so-called \textit{depolarization} process\cite{MAG06,Sho08}:
Alice and Bob generate three independent and uniformly random bits $\alpha$, $\beta$ and $\gamma$, and modify their inputs and outputs by 
\beba
x &\to x+\alpha, \\
y &\to y+\beta, \\
a &\to a+\beta x+ \alpha \beta+\gamma,\\
b &\to b+\alpha y+\gamma.
\eeea
Here we prove an analogue of this result for $p$-nary-input, $p$-nary-output boxes.

\begin{definition}
\be
PR_{p,j}(a,b|x,y)=
\bc
\dfrac{1}{p},~~~~\textrm{if}~~a-b=xy-j\\
0,~~~~\textrm{otherwise}
\ec,
\ee
where $a,b,x,y,j \in \zp$.
\end{definition}

\begin{lemma}
Given any $p$-nary-input, $p$-nary-output box $P$, define
\beba
\mu_j=\dfrac{1}{p^2}\ssll{x,y=0}{p-1}P(a-b=xy-j|x,y),
\label{eq:muj}
\eeea
where
\beba
P(a-b=xy-j|x,y)=\ssll{k=0}{p-1}P(a=k,b=k-xy+j|x,y).
\eeea
Then
\be
P \to \sum\limits_{j=0}^{p-1}\mu_j PR_{p,j}.
\ee
\label{lem:dep}
\end{lemma}
\begin{proof}
Consider the following protocol: Alice and Bob generate three independent and uniformly random variable $\alpha,\beta,\gamma \in \zp$. They input
\beba
x'&=x+\alpha,\\
y'&=y+\beta
\eeea
to box $P$. Suppose they obtain outputs $a'$ and $b'$. Then they set
\beba
a&=a'-\beta x-\alpha\beta+\gamma,\\
b&=b'+\alpha y+\gamma
\eeea 
as their final outputs.

Suppose this protocol realizes a box $\widetilde{P}$.
Then for any given $a,b,x,y$,
\beba
\widetilde{P}(a,b|x,y)=\dfrac{1}{p^3}\ssll{\alpha,\beta,\gamma=0}{p-1}P(a',b'|x',y').
\eeea
Note that
\beba
a'-b'-x'y'&=&(a+\beta x+\alpha\beta-\gamma)-(b-\alpha y-\gamma)\\
&&-(x+\alpha)(y+\beta)\\
&=&a-b-xy.
\eeea
Besides, $(x',y')$ are uniformly random in $\zp \times \zp$; $a'$ (or $b'$) is also uniformly random in $\zp$ conditioned on any $(x',y')$. So, if $a-b=xy-j$, then
\beba
\widetilde{P}(a,b|x,y)&=&\dfrac{1}{p^3}\ssll{x',y'=0}{p-1}P(a'-b'=x'y'-j|x',y')\\
&=&\dfrac{\mu_j}{p},
\eeea
which implies
\beba
\widetilde{P}=\ssll{j=0}{p-1}\mu_jPR_{p,j}.
\eeea
\end{proof}

Hence, any $p$-nary-input, $p$-nary-output box can be transformed into a probabilistic mixture of $PR_{p,j}$'s, each of which is essentially equivalent to $PR_{p}$ up to an additive shift of the outputs.

\begin{remark}
Lemma \ref{lem:dep} can be straightforwardly generalized to the case of any $q$-nary input/output box where $q$ does not have to be prime.
\end{remark}

\subsection{Limits on Nonlocality from Information Causality}

Let us briefly review the principle of information causality\cite{PPK+09}. It was introduced via the following communication task, which is similar to a random access code \cite{AN+02} or oblivious transfer \cite{Rab81,WW05}. Suppose Alice and Bob are two spatially separated parties. Alice receives a string of $N$ random and independent $p$-nary digits $\vx=(x_0,x_1,\dots,x_{N-1}) \in \zpN$. Bob receives a random variable $y \in \{0,1,\dots,N-1\}$ and is asked to give the $y$-th digit of Alice. To achieve this, they may share in advance some no-signalling resources such as shared randomness, entangled states or nonlocal boxes. Besides, Alice is allowed to send at most $m$ $p$-nary digits (or equivalently, $m\log p$ bits) to Bob. Let us denote Bob's output by $b$. The degree of their success is quantified by
\beba
I  \equiv \ssll{i=0}{N-1}I(x_i:b|y=i)
\eeea
where $I(x_i:b|y=i)$ is the Shannon mutual information between $x_i$ and $b$, under the condition that Bob receives $y=i$. Note that if $P(b=x_i|y=i)=p_i$, then by Fano's inequality,
\beba
I \ge N \log p -\ssll{i=0}{N-1} h(p_i)
-\ssll{i=0}{N-1}(1-p_i)\log (p-1)
\label{eq:fano}
\eeea
where $h(x)=-x\log x-(1-x)\log (1-x)$ is the binary entropy function.

The principle of information causality states that for any physically allowed theories, we must have
\beba
I \le m\log p.
\label{eq:ic}
\eeea
Both classical and quantum correlations satisfy this condition. However, it is unknown whether all postquantum correlations violate this condition.

Now suppose Alice and Bob share unlimited number of copies of a $p$-nary-input, $p$-nary-output box $P$, and Alice is allowed to send only one $p$-nary digit to Bob, i.e. $m=1$. We are going to investigate how $P$ can help them in this task, and presents several limits on $P$ from condition (\ref{eq:ic}).

We will consider the cases of $N \le p$ and $N \ge p$ separately.

\subsubsection*{Case 1: $N \le p$}
In this case, Alice receives $\vx=(x_1,x_2,\dots,x_N) \subseteq \zpN$
and Bob receives $y \in \{0,1,\dots,N-1\}\subseteq \zp$.
Bob aims to obtain the value of
\beba
F(\vx,y) \equiv x_y.
\eeea
Note that
\beba
F(\vx,y)=\ssll{i=0}{N-1}[\ppll{0 \le j \neq i \le N-1}{}(i-j)^{-1}(y-j)]x_i
\label{eq:Fvxy1}
\eeea
for any $\vx \in \zpN$ and $y \in \{0,1,\dots,N-1\}\subseteq \zp$.
Moreover, the right-hand side of the above equation has degree $N-1$ in $y$. So $F(\vx,y)$ can be rewritten as
\beba
F(\vx,y)=\ssll{k=0}{N-1} y^k F_k(\vx)
\label{eq:Fvxy2}
\eeea
for some $F_k:\zpN \to \zp$. This fact suggests a protocol that generates each term $y^k F_k(\vx)$ independently and then sums them together. For $PR_p$, this can be achieved by Alice inputting $F_k(\vx)$ and Bob inputting $y^k$, then the difference between their outputs would be $y^k F_k(\vx)$. For a general box $P$, by lemma \ref{lem:dep}, it can be converted into a mixture of $PR_{p,j}$ boxes which can be viewed as an imperfect $PR_p$ box with random additive noise. As long as the random noise is not very bad, we can still pretend it to be a $PR_p$ and achieve a high efficiency. So consider the following protocol:\\\\

\begin{protocol}[htbp]
\caption{~~~basicRAC$(p,N,c,P,\vx,y)$}
\begin{tabular}{p{0.06 \textwidth}p{0.4 \textwidth}}
\textbf{Setup:} & $p$ is prime. $N \le p$. Alice has $\vx \in \zpN$ and Bob has $y \in \{0,1,\dots,N-1\}$. They share at least
$N-1$ copies of a $p$-nary-input, $p$-nary-output box $P$.
They also choose $c \in \zp$.\\
\textbf{Steps:} &
1. Alice and Bob convert each copy of $P$ into a copy of $\widetilde{P}= \ssll{j=0}{p-1}\mu_j PR_{p,j}$ (where $\mu_j$ is given by Eq.(\ref{eq:muj})), using the protocol given in the proof of lemma \ref{lem:dep}.\\
&2. Alice and Bob use $N-1$ copies of box $\widetilde{P}$ as follows: for the $k$-th box, Alice inputs $F_k(\vx)$ (which is given by Eqs.(\ref{eq:Fvxy1}) and (\ref{eq:Fvxy2})) and Bob inputs $y^k$, and suppose they get outputs
$a_k$ and $b_k$ respectively, $\forall k=1,2, \dots,N-1$.\\
&3. Alice sends $q=\ssll{k=1}{N-1} a_k + F_0(\vx)$ to Bob.\\
&4. After receiving $q$, Bob outputs $b=q-\ssll{k=1}{N-1} b_k-c$.\\
\end{tabular}
\end{protocol}

For an illustration of this protocol, see Fig. \ref{fig:basic}.

\begin{figure}[htbp]
\center
\includegraphics[width=0.9 \columnwidth, height=\columnwidth]{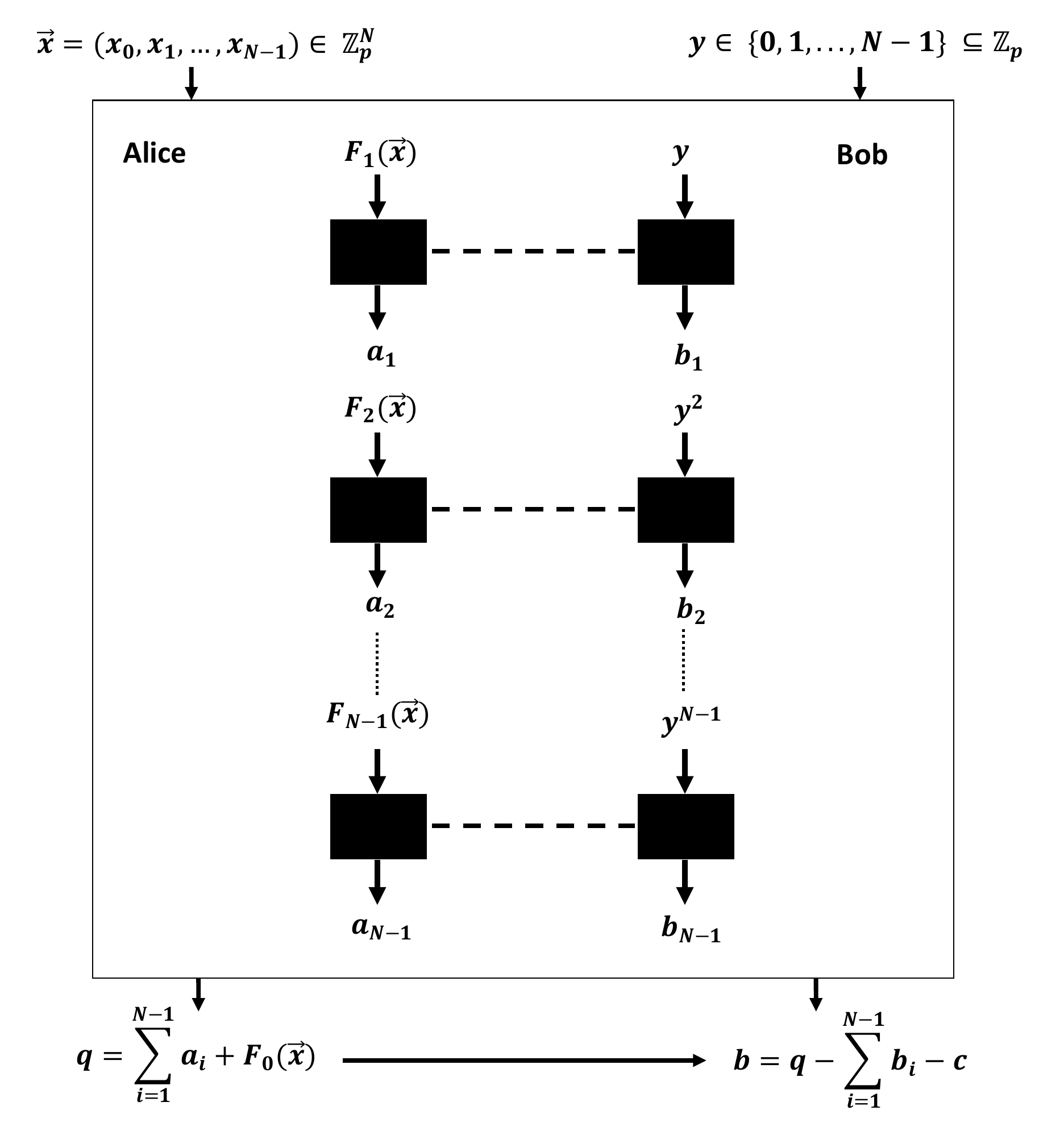}\\
\caption{Basic Protocol for $N\le p$. Each black box represents a copy of $\widetilde{P}$ that is obtained from a copy of $P$ using the protocol given in the proof of lemma \ref{lem:dep}. $F_0(\vx),F_1(\vx),\dots,F_{N-1}(\vx)$ are given by Eqs.(\ref{eq:Fvxy1}) and (\ref{eq:Fvxy2}). }
\label{fig:basic}
\end{figure}

\begin{remark}
When $p=2$ and $N=2$, we have $F(x_0,x_1,y)=x_0+(x_0+x_1)y$ and hence $F_1(x_0,x_1)=x_0+x_1$, $F_0(x_0,x_1)=x_0$. Then the above protocol with $c=0$ reduces to the basic protocol in Ref.\cite{PPK+09}.
\end{remark}

Let us analyse the efficiency of this protocol.

First, consider the special case
of $P=PR_p$ and $c=0$. In this case, step 1 does not have any effect on $PR_p$, and we have $\widetilde{P}=PR_p$. So in step 2, we have
\beba
a_k-b_k=F_k(\vx)y^k, ~~~\forall k.
\eeea
As a result, Bob's output 
\beba
b=q-\ssll{k=1}{N-1} b_k
=\ssll{k=1}{N-1} (a_k-b_k) + F_0(\vx)
=F(\vx,y).
\eeea 
So Bob correctly guesses $x_y$ with certainty.

However, by lemma \ref{lem:dep}, $\widetilde{P}$ is generally not $PR_p$, but equals $PR_{p,j}$ with probability $\mu_j$, $\forall j\in \zp$. Then 
\beba
a_k-b_k=F_k(\vx)y^k-j_k
\eeea
where 
\beba
P(j_k=j)=\mu_j,~~~\forall j\in \zp.
\eeea
And Bob's output is
\beba
b=F(\vx,y)-\ssll{k=1}{N-1}j_k-c,
\eeea
which is correct if and only if 
\beba
\ssll{k=1}{N-1}j_k=-c,
\eeea
which happens with probability
\beba
\chi_p(\vec{\mu},N-1,c) \equiv
\ssll{j_1+\dots+j_{N-1}=-c}{}\ppll{k=1}{N-1}{\mu_{j_k}},
\eeea
where $\vec{\mu}=(\mu_0,\mu_1,\dots,\mu_{p-1})$. Note that this probability is independent of $y$.
So by Eq.(\ref{eq:fano}), we must have
\beba
\chi_p(\vec{\mu},N-1,c) \le \tpi{\dfrac{N-1}{N}\log p},
\eeea
where
\be
\tp{x}=h(x)+(1-x)\log (p-1)
\ee
and $\tpi{\cdot}$ is its inverse function
\footnote{For any $p$, the function $\tp{x}$ first monotonically increases, then monotonically drops, as $x$ grows from $0$ to $1$. Only the first part is interesting to us, and $\tpi{\cdot}$ is the inverse function for that part.}, for condition (\ref{eq:ic}) to be satisfied.

Note that $\forall c \in \zp$, $\forall M \ge 1$,
\beba
\chi_p(\vec{\mu},M,c)
&=&\ssll{j_1+\dots+j_{M}=-c}{}\ppll{k=1}{M}{\mu_{j_k}}\\
&=&\dfrac{1}{p}\textrm{tr}(Z_p^{c}(\ssll{j=0}{p-1}\mu_j Z^j_p)^{M})\\
&=&\dfrac{1}{p}\ssll{k=0}{p-1}\omega_p^{ck}(\ssll{j=0}{p-1}\mu_j\omega_p^{jk})^{M}.
\eeea
where $Z_p=\textrm{diag}(1,\omega_p,\dots,\omega_p^{p-1})$
is the $p$-dimensional generalization of Pauli $Z$ matrix,
$\omega_p=e^{i{2\pi}/{p}}$, and in the second step we use $\textrm{tr}(Z^i_p)=0$, $\forall i=1,2,\dots,p-1$. Therefore,

\begin{theorem}
In any world where information causality holds,
any $p$-nary-input, $p$-nary-output box $P$ satisfies:
$\forall N \le p$, $\forall c \in \zp$,
\beba
\dfrac{1}{p}\ssll{k=0}{p-1}\omega_p^{ck}(\ssll{j=0}{p-1}\mu_j\omega_p^{jk})^{N-1} \le \tpi{\dfrac{N-1}{N}\log p},
\eeea
where $\mu_j$ is defined as Eq.(\ref{eq:muj}).
\label{thm:icpr}
\end{theorem}

\subsubsection*{Case 2: $N \ge p$} 

Assume $N=p^n$ and $y=\ssll{i=0}{n-1}y_ip^i$ for $y_i \in \zp$. We present a recursive protocol that calls the basicRAC protocol as a subroutine. It is the $p$-dimensional generalization of the one given in Ref.\cite{PPK+09}.

\begin{protocol}[htb]
\caption{~~~recursiveRAC$(p,N,c,P,\vx,y)$}
\begin{tabular}{p{0.06 \textwidth}p{0.4 \textwidth}}
\textbf{Setup:} & $p$ is prime. $N=p^n$ for some $n \ge 1$. Alice has $\vx \in \zpN$ and Bob has $y=\ssll{i=0}{n-1}y_ip^i$ where $y_i\in \zp$.
They share at least $N-1$ copies of a $p$-nary-input, $p$-nary-output box $P$. They also choose $c \in \zp$.\\
\textbf{Steps:} & If $N=p$, then Alice and Bob execute basicRAC$(p,N,c,P,\vx,y)$. Otherwise:
\begin{itemize}
\item \textbf{Alice}: she divides her input $\vx$ into $\frac{N}{p}$ substrings:
$\vz_0\equiv(x_0,x_1,\dots,x_{p-1})$,
$\vz_1\equiv(x_p,x_{p+1},\dots,x_{2p-1})$,
$\dots$, $\vz_{\frac{N}{p}-1} \equiv(x_{N-p},x_{N-p+1},\dots,x_{N-1})$.
For $k=0,1,\dots,\frac{N}{p}-1$, she executes her part of basicRAC$(p,p,c,P,\vz_k,y_0)$, except that she does not send her message (which is denoted by $q_k$). Then she executes her part of recursiveRAC$(p,\frac{N}{p},0,P,\vec{q},y')$,
where $\vec{q}=(q_0,q_1,\dots,q_{\frac{N}{p}-1})$ and $y'=\ssll{i=1}{n-1}y_ip^{i-1}$.
\item \textbf{Bob}: he executes his part of recursiveRAC$(p,\dfrac{N}{p},0,P,\vec{q},y')$. (Here still $y'=\ssll{i=1}{n-1}y_ip^{i-1}$.) Suppose the output from this protocol is $\hat{q}$. Then he executes his part of basicRAC$(p,p,c,\vz_{y'},y_0)$, except that he uses $\hat{q}$ as the message received from Alice.  The output from this protocol is set as his final output.
\end{itemize}
\end{tabular}
\end{protocol}

Fig. \ref{fig:recursive} illustrates an example of this protocol for $p=3$ and $N=9$.

\begin{figure}[htbp]
\center
\includegraphics[width=0.9 \columnwidth, height=1.2 \columnwidth]{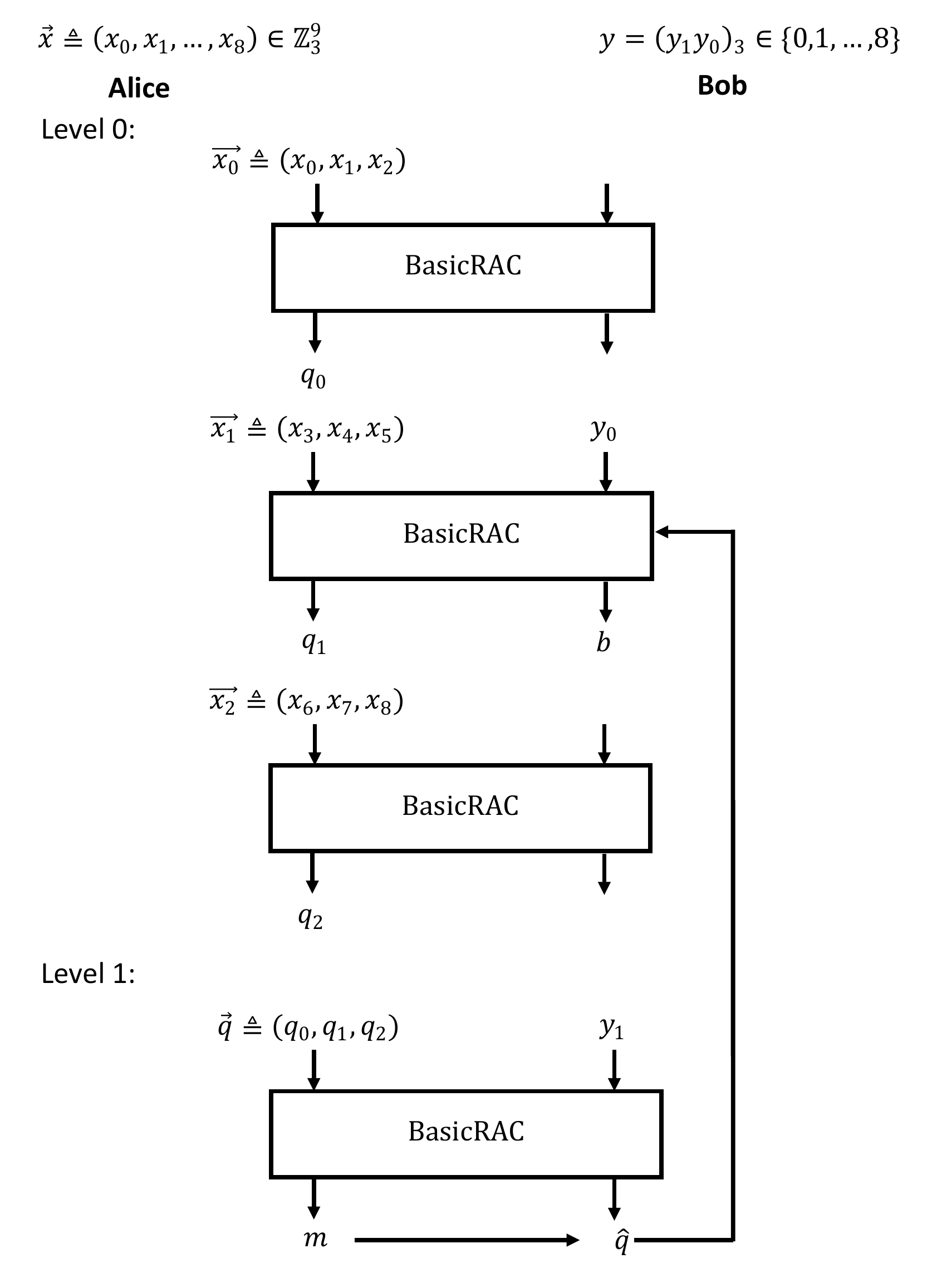}\\
\caption{Recursive protocol for $N \ge p$. Here we show an example for $p=3$ and $N=9$.
Suppose Alice receives $\vx=(x_0,x_1,\dots,x_8)\in \mathbb{Z}^9_3$
and Bob receives $y=(y_1y_0)_3 \in \{0,1,\dots,8\}$. They also share at least $N-1=8$
copies of some box $P$. The protocol consists of three level-$0$ and one level-$1$
executions of the basic protocol illustrated in Fig.\ref{fig:basic}. For Alice, she executes
the three level-$0$ basic protocols with inputs $(x_0,x_1,x_2)$, $(x_3,x_4,x_5)$ and $(x_6,x_7,x_8)$
respectively. Suppose her messages (which are not really sent) are $q_0,q_1,q_2$ respectively.
Then she executes the level-$1$ basic protocol with input $(q_0,q_1,q_2)$,
and truly sends her message for this one. For Bob, he executes the level-$1$ basic protocol
with input $y_1$, and suppose his output is $\widehat{q}$. Then he executes only the $y_1$-th level-$0$ basic protocol with input $y_0$(this figure shows the case $y_1=1$). During this protocol he
pretends that $\widehat{q}$ is the message received from Alice. At last, this basic protocol's output
$b$ is set as Bob's final output. Note that each basic protocol costs $p-1=2$ copies of box $P$, so
Alice totally accesses $8$ copies of $P$, while Bob only accesses $4$ copies of $P$.}
\label{fig:recursive}
\end{figure}

The recursiveRAC protocol can be viewed as a level-$n$ pyramid of basicRAC protocols. Its idea is that Alice uses level $k+1$ to transmit her messages generated at level $k$, while Bob uses level $k+1$ to reveal the message he needs
at level $k$, $\forall k=0,1,\dots,n-1$. Although Alice accesses totally $(p-1)\ssll{k=0}{n-1}p^k=N-1$ copies of $\widetilde{P}$, Bob only accesses
$n(p-1)$ copies of $\widetilde{P}$ and only these boxes are truly relevant to his final output. Each of these boxes contributes to his final output an additive shift that equals $j$ with probability $\mu_j$,
$\forall j \in \zp$. So Bob's final guess is correct if and only if all these additive shifts sum to $-c$, which happens with probability
\beba
\chi_p(\vec{\mu},n(p-1),c) &=& \ssll{j_1+\dots+j_{n(p-1)}=-c}{}\ppll{k=1}{n(p-1)}{\mu_{j_k}}\\
&=&\dfrac{1}{p}\ssll{k=0}{p-1}\omega^{ck}(\ssll{j=0}{p-1}\mu_j\omega^{jk})^{n(p-1)}.
\eeea
So, by Eq.(\ref{eq:fano}), we must have
\beba
\chi_p(\vec{\mu},n(p-1),c) \le \tpi{\dfrac{N-1}{N}\log p},
\eeea
otherwise condition (\ref{eq:ic}) is violated. Thus,

\begin{theorem}
In any world where information causality holds,
any $p$-nary-input, $p$-nary-output box $P$ satisfies:
$\forall n\ge 1$, $N=p^n$, $\forall c \in \zp$,
\beba
\dfrac{1}{p}\ssll{k=0}{p-1}\omega_p^{ck}(\ssll{j=0}{p-1}\mu_j\omega_p^{jk})^{n(p-1)} \le \tpi{\dfrac{N-1}{N}\log p},
\eeea
where $\mu_j$ is defined as Eq.(\ref{eq:muj}).
\label{thm:icpr2}
\end{theorem}

\subsubsection{Bounds with respect to general functional boxes}

Theorems \ref{thm:icpr} and \ref{thm:icpr2} can be viewed as giving bounds on the proximity between a plausible $p$-nary-input, $p$-nary-output box and $PR_p$. By lemma \ref{lem:pf}, $P^f$ and $PR_p$ are interconvertible for any additively inseparable function $f$. So there should be also bounds on the proximity between a plausible $p$-nary-input, $p$-nary-output box and the corresponding functional box. In what follows, we will give several such bounds.

\begin{definition}
For any function $f:\zp \times \zp \to \zp$, define
\be
P^f_j(a,b|x,y)=
\bc
\dfrac{1}{p},~~~~\textrm{if}~~a-b=f(x,y)-j\\
0,~~~~\textrm{otherwise}
\ec,
\ee
where $a,b,x,y,j \in \zp$.
\end{definition}
$P^f_j$ and $P^f$ are essentially equivalent, except for an additive shift of the outputs.

We will consider the case of $\Delta(f)=2$ and $\Delta(f)>2$ separately.\\

\textbf{Case 1}: $\Delta(f)=2$.

Suppose
\beba
f(x,y)=\lambda xy+g(x)+h(y)
\eeea
for some $\lambda \neq 0$ and $g,h:\zp \to \zp$.

Given any $p$-nary-input, $p$-nary-output box $P$, define
\beba
\nu_j=\dfrac{1}{p^2}\ssll{x,y=0}{p-1}P(a-b=f(x,y)-j|x,y)
\label{eq:nuj}
\eeea
where
\beba
&&P(a-b=f(x,y)-j|x,y)\\
&=&\ssll{k=0}{p-1}P(a=k,b=k-f(x,y)+j|x,y).
\eeea
Then we have
\beba
P \to \ssll{j=0}{p-1}\nu_j P^f_j.
\label{eq:ppfj}
\eeea
To prove this, consider the following protocol:
Alice and Bob generate three independent and uniformly random variable $\alpha$, $\beta$, $\gamma \in \zp$.
Alice inputs 
\beba
x'=x+\alpha
\eeea
to $P$, and Bob inputs 
\beba
y'=y+\beta
\eeea 
to $P$. Suppose they receive outputs $a'$ and $b'$ respectively. Then Alice sets
\beba
a=a'-\lambda \beta x -\lambda\alpha\beta -g(x+\alpha)+g(x)
+\gamma
\eeea 
as her final output,
and Bob sets 
\beba
b=b'+\lambda \alpha y+h(y+\beta)-h(y)+\gamma
\eeea
as his final output. Suppose this protocol realizes a box $\widehat{P}$.
Then for any given $a,b,x,y$,
\beba
\widehat{P}(a,b|x,y)=\dfrac{1}{p^3}\ssll{\alpha,\beta,\gamma=0}{p-1}P(a',b'|x',y').
\eeea
Note that
\beba
a'-b'-f(x',y')&=&(a+\lambda \beta x +\lambda\alpha\beta +g(x+\alpha)-g(x)\\
&&-\gamma)-(b-\lambda \alpha y-h(y+\beta)+h(y)\\
&&-\gamma)-(\lambda (x+\alpha)(y+\beta)+g(x+\alpha)\\
&&+h(y+\beta))\\
&=&a-b-f(x,y).
\eeea
Besides, $(x',y')$ are uniformly random in $\zp \times \zp$; $a'$ (or $b'$) is also uniformly random in $\zp$ conditioned on any $(x',y')$. So, if $a-b=f(x,y)-j$, then
\beba
\widehat{P}(a,b|x,y)&=&\dfrac{1}{p^3}\ssll{x',y'=0}{p-1}P(a'-b'=f(x',y')-j|x',y')\\
&=&\dfrac{\nu_j}{p},
\eeea
which implies
\beba
\widehat{P}=\ssll{j=0}{p-1}\nu_jP^f_j.
\eeea

Now recall that in the proof of lemma {\ref{lem:pf}} we gave the following protocol that converts a $P^f$ to a $PR_p$: Alice inputs $x$ and Bob inputs $y$, and suppose they receive $a$ and $b$. Their final outputs are $a'=\lambda^{-1}(a-g(x))$ and $b'=\lambda^{-1}(b+h(y))$. Note
\beba
a'-b'-xy=\lambda^{-1}(a-b-f(x,y)).
\eeea
So this protocol converts a $P^f_j$ to a $PR_{p,\lambda^{-1}j}$. Hence, we have
\beba
\ssll{j=0}{p-1}\nu_jP^f_j \to \ssll{j=0}{p-1}\nu_jPR_{p,\lambda^{-1}j}.
\label{eq:pfjprpj}
\eeea

Combining Eq.(\ref{eq:ppfj}) and (\ref{eq:pfjprpj}),
we obtain
\beba
P \to \ssll{j=0}{p-1}\nu_j PR_{p,\lambda^{-1}j}.
\eeea
Then by theorems \ref{thm:icpr} and
\ref{thm:icpr2}, we get

\begin{theorem}
In any world where information causality holds,
for any additively inseparable function $f:\zp \times \zp \to \zp$ with $\Delta(f)=2$ and any $p$-nary-input, $p$-nary-output box $P$, we have:
\begin{itemize}
\item $\forall N \le p$, $\forall c \in \zp$,
\beba
\dfrac{1}{p}\ssll{k=0}{p-1}\omega_p^{ck}(\ssll{j=0}{p-1}\nu_j\omega_p^{jk})^{N-1} \le \tpi{\dfrac{N-1}{N}\log p};
\eeea
\item $\forall n \ge 1$, $N=p^n$, $\forall
c \in \zp$,
\beba
\dfrac{1}{p}\ssll{k=0}{p-1}\omega_p^{ck}(\ssll{j=0}{p-1}\nu_j\omega_p^{jk})^{n(p-1)} \le \tpi{\dfrac{N-1}{N}\log p},
\eeea
where $\nu_j$ is defined as Eq.(\ref{eq:nuj}).
\end{itemize}
\label{thm:icpf1}
\end{theorem}

\textbf{Case 2}: $\Delta(f)>2$.

Given arbitrary $p$-nary-input, $p$-nary-output box $P$, we change it slightly as follows: Alice and Bob generate a uniformly random $\gamma \in \zp$, and modify their outputs by $a \to a+\gamma$ and $b \to b+\gamma$. Suppose this modified box is $P'$. They we have
\beba
P'(a,b|x,y)=\dfrac{1}{p}E_{x,y,f(x,y)-a+b}
\eeea
where
\beba
E_{x,y,j} & = & P(a-b=f(x,y)-j|x,y)\\
&=& \ssll{k=0}{p-1}P(a=k,b=k-f(x,y)+j|x,y).
\label{eq:exyj}
\eeea
Define
\beba
\nu_j=\min\limits_{x,y \in \zp}{E_{x,y,j}}.
\label{eq:nuj2}
\eeea
Then we have
\beba
P'=\ssll{j=0}{p-1} \nu_j P^f_j + (1-\ssll{j=0}{p-1} \nu_j)P''
\eeea
for some box $P''$.

Now recall that in the proof of lemma \ref{lem:pf},
we have given a recursive protocol that transforms $M \equiv 2^{\Delta(f)-2}$ copies of $P^f$ into a copy of $PR_p$. Let us see what happens if we apply that protocol to $P'$. Suppose the resulting box is $\widetilde{P}$.
Note that $P'$ acts as $P^f_j$ with probability $\nu_j$ (or acts as $P''$ with probability $1-\ssll{j=0}{p-1}\nu_j$). So each time Alice and Bob access a $P'$, their outputs are basically the same as those of $P^f$ except for an additive shift which is $j$ with probability $\nu_j$ (or with probability $1-\ssll{j=0}{p-1}\nu_j$ the outputs are nonsense). These additive shifts are multiplied by a factor $\lambda^{-1}$ for some $\lambda\neq 0$ (see Eq.(\ref{eq:pfjprpj})). And the overall shift in the final output is the sum of a half of these multiplied shifts minus the sum of the other half, since we use a difference method. So,
\beba
\widetilde{P} = \ssll{j=0}{p-1}{\mu_j} PR_{p,j} +(1-\ssll{j=0}{p-1}\mu_j)P''',
\eeea
where
\beba
\mu_j &=& \ssll{(j_1,\dots,j_M) \in \mathcal{S}_{M,j}} {}
\ppll{k=1}{M}\nu_{j_k}
\label{eq:mujnuj}
\eeea
in which
\beba
\mathcal{S}_{M,j}\equiv\{(j_1,\dots,j_M)\in \mathbb{Z}^M_p:
\ssll{k=1}{\frac{M}{2}}j_k - \ssll{k=\frac{M}{2}+1}{M}j_k=\lambda j\}
\label{eq:smj}
\eeea
and $P'''$ is some box.

Then, by lemma \ref{lem:dep}, we have
\beba
\widetilde{P} \to \widehat{P}=\ssll{j=0}{p-1}\widehat{\mu_j}PR_{p,j}
\eeea
for some $\widehat{\mu_j}\ge \mu_j$.

Now assume Alice and Bob execute basicRAC$(p,N,c,\widehat{P},\vx,y)$ for $N \le p$. Then Bob's guess is correct with probability
\beba
\chi_p(\vec{\widehat{\mu}},N-1,c)&=& \ssll{j_1+\dots+j_{N-1}=-c}{}\ppll{k=1}{N-1}{\widehat{\mu_{j_k}}}\\
&\ge &
\ssll{j_1+\dots+j_{N-1}=-c}{}\ppll{k=1}{N-1}{\mu_{j_k}}\\
&=&
\ssll{(j_1,\dots,j_{M(N-1)}) \in \mathcal{S}_{M(N-1),-c}}{}\ppll{k=1}{M(N-1)}{\nu_{j_k}}\\
& \equiv & \sigma_p (\vec{\nu}, M(N-1), c),
\eeea
where in the second step we use $\widehat{\mu_j} \ge \mu_j \ge 0$, and in the third step
we use Eqs.(\ref{eq:mujnuj}) and (\ref{eq:smj}), and in the last step
$\vec\nu=(\nu_0,\nu_1,\dots,\nu_{p-1})$.
So we must have
\beba
\sigma_p (\vec{\nu}, M(N-1), c) \le \tpi{\dfrac{N-1}{N}\log p},
\eeea
otherwise condition (\ref{eq:ic}) is violated.

Similarly, if Alice and Bob execute recursiveRAC$(p,N,c,\widehat{P},\vx,y)$ for $N=p^n$, then Bob's guess is correct with probability
\beba
\chi_p(\vec{\widehat{\mu}},n(p-1),c)
\ge \sigma_p (\vec{\nu}, Mn(p-1), c),
\eeea
So unless
\beba
\sigma_p (\vec{\nu}, Mn(p-1), c) \le \tpi{\dfrac{N-1}{N}\log p},
\eeea
condition (\ref{eq:ic}) is violated.

Note that $\forall L \ge 1$, $\forall c \in \zp$,
\beba
\sigma_p (\vec{\nu}, L, c)
&=&
\ssll{(j_1,\dots,j_L) \in \mathcal{S}_{L,-c}}{}\ppll{k=1}{L}{\nu_{j_k}}\\
&=& \dfrac{1}{p}\mathrm{tr}(Z_p^{\lambda c} (\ssll{j=0}{p-1}{\nu_j Z^j_p})^{\frac{L}{2}}
(\ssll{j=0}{p-1}{\nu_j Z^{-j}_p})^{\frac{L}{2}})\\
&=& \dfrac{1}{p} \ssll{k=0}{p-1}
\omega_p^{\lambda c k}
(\ssll{j=0}{p-1}{\nu_j \omega_p^{jk}})^{\frac{L}{2}}
(\ssll{j=0}{p-1}{\nu_j \omega_p^{-jk}})^{\frac{L}{2}},
\eeea
where in the second step we use $\textrm{tr}(Z^i_p)=0$, $\forall i=1,2,\dots,p-1$. So we have

\begin{theorem}
In any world where information causality holds,
for any additively inseparable function $f:\zp \times \zp \to \zp$ with $\Delta(f)>2$ and any
$p$-nary-input, $p$-nary-output box $P$, we have:
\begin{itemize}
\item $\forall N \le p$, $L=(N-1)2^{\Delta(f)-2}$,
$\forall c \in \zp$,
\beba
\dfrac{1}{p} \ssll{k=0}{p-1}
\omega_p^{ c k}
(\ssll{j=0}{p-1}{\nu_j \omega_p^{jk}})^{\frac{L}{2}}
(\ssll{j=0}{p-1}{\nu_j \omega_p^{-jk}})^{\frac{L}{2}} \le \tpi{\dfrac{N-1}{N}\log p};
\eeea
\item $\forall n \ge 1$, $N=p^n$, $L=n(p-1)2^{\Delta(f)-2}$,
$\forall c \in \zp$,
\beba
\dfrac{1}{p} \ssll{k=0}{p-1}
\omega_p^{ c k}
(\ssll{j=0}{p-1}{\nu_j \omega_p^{jk}})^{\frac{L}{2}}
(\ssll{j=0}{p-1}{\nu_j \omega_p^{-jk}})^{\frac{L}{2}} \le \tpi{\dfrac{N-1}{N}\log p},
\eeea
where $\nu_j$ is defined as Eqs.(\ref{eq:exyj}) and (\ref{eq:nuj2}).
\end{itemize}
\label{thm:icpf2}
\end{theorem}

\begin{remark}
Since quantum correlations satisfy the principle of
information causality, theorems \ref{thm:icpr}, \ref{thm:icpr2}, 
\ref{thm:icpf1} and \ref{thm:icpf2} all apply to them.
\end{remark}

\section{Conclusion}

%summary
In summary, we have proposed the class of functional boxes which incorporate the generalized PR boxes as a special case. Every functional box corresponding to an additively inseparable function is asymptotically equivalent to a generalized PR box, which can enable perfect distributed computation and make communication complexity trivial. So all such functional boxes are unlikely to exist. Furthermore, we investigated how proximate can a general box be to these functional boxes without violating the principle of information causality.

%open questions
Our work raises many new questions:

%probabilistic communication complexity
First, we have shown that if $PR_p$ box can be implemented exactly, it would lead to the collapse of \emph{deterministic} communication complexity. But we do not know how much noise it can tolerate while still making \emph{probabilistic} communication complexity trivial. And what about general $P^f$?

% Pf to PR? optimal?
Second, in the proof of lemma \ref{lem:pf}, we gave a protocol that
transforms $2^{\Delta(f)-2}$ copies of $P^f$ into a $PR_p$. That protocol is universal, but might be not optimal for some $f$. If we can simulate a $PR_p$ with fewer copies of $P^f$, then the bounds in theorem \ref{thm:icpf2} can be improved accordingly. In fact,
can we directly use $P^f$ to perform distributed computation, instead of first converting it into $PR_p$?

% closed set
Third, as pointed out in Ref.\cite{ABL+09}, the set of physically allowed boxes should form a closed set under local wirings. Namely, if a set of boxes $P_1,P_2,\dots,P_m$ are all allowed by a physical theory, then a new box $\widehat{P}$ obtained by locally connecting these boxes should also be allowed by this theory. Conversely, if $\widehat{P}$ is not allowed, then at least one of $P_1,P_2,\dots,P_m$ should be prohibited. Since we have already obtained a set of implausible postquantum correlations, can we use this approach to rule out more?

% general setting
Finally, here we have only considered bipartite $p$-nary-input, $p$-nary-output boxes. It would be interesting to extend our results to more general boxes with arbitrary number of inputs and outputs, as well as multipartite boxes.

\section*{Acknowledgments}
This research was supported by NSF Grant CCR-0905626
and ARO Grant W911NF-09-1-0440.

\end{document}